\documentclass[conference,a4paper]{IEEEtran}

\usepackage[utf8]{inputenc} 
\usepackage[T1]{fontenc}
\usepackage{url}
\usepackage{ifthen}
\usepackage{cite}
\usepackage[cmex10]{amsmath} 
\usepackage{comment}
\usepackage{mathrsfs}
\usepackage{amsfonts,dsfont,amssymb,amsthm,stmaryrd,bbm,color}
\usepackage{color}
\usepackage{graphicx}
\newtheorem{conj}{Conjecture}[section]
\newtheorem{thm}{Theorem}[section]
\newtheorem{rmk}[conj]{Remark}
\newtheorem{lem}[conj]{Lemma}
\newtheorem{prop}[conj]{Proposition}

\newtheorem{defn}[conj]{Definition}

\newcommand{\R}{\mathbb{R}} 

\newcommand{\Z}{\mathbb{Z}}

\def \O{{\mathcal{O}}}

\begin{document}
\title{Error Bounds on a Mixed Entropy Inequality} 


\author{%
  \IEEEauthorblockN{James Melbourne, Saurav Talukdar, Shreyas Bhaban and Murti V. Salapaka}
  \IEEEauthorblockA{University of Minnesota,
                    Minneapolis, USA \\
                    Email: \{melbo013, taluk005, bhaba001, murtis\}@umn.edu}
}


\maketitle

\begin{abstract}
   Motivated by the entropy computations relevant to the evaluation of decrease in entropy in bit reset operations, the authors investigate the deficit in an entropic inequality involving two independent random variables, one continuous and the other discrete.  In the case where the continuous random variable is Gaussian, we derive strong quantitative bounds on the deficit in the inequality.  More explicitly it is shown that the decay of the deficit is sub-Gaussian with respect to the reciprocal of the standard deviation of the Gaussian variable.  What is more, up to rational terms these results are shown to be sharp.
\end{abstract}


\section{Introduction}
Gaussian mixtures are an interesting class of probability distributions arising in a multitude of disciplines like machine learning \cite{bilmes1998gentle, rasmussen2006gaussian}, signal processing\cite{kostantinos2000gaussian}, thermodynamics of information \cite{parrondo2015thermodynamics} and many more. Hence, entropy of Gaussian mixtures is of great importance. However, analytical expression for the entropy of Gaussian mixtures is not available and researchers sometimes resort to numerical approximations as a substitute \cite{michalowicz2008calculation}. 

In this article we present sharp bounds on the entropy of Gaussian mixtures. We arrive at Gaussian mixtures as the density of the sum of a continuous and discrete random variable denoted by $X$ and $Z$ respectively, with $X$ being Gaussian and independent of $Z$. It is shown later that the density of $X+Z$ is a mixture of Gaussian densities. The mixed random variable $X+Z$ is to be interpreted as- for each discrete value taken by $Z$, there is an associated Gaussian density around it, that is, if we observe $X+Z$, we will obtain the discrete values corrupted with Gaussian noise and it is indeed reasonable to assume that the Gaussian noise associated with a particular realization of $Z$ is independent of the value taken by $Z$. A realistic example of this is in intracellular transport, where nano-molecular machines referred as molecular motors \cite{schliwa2003molecular} transport important \lq cargoes\rq \ inside the cell from one location to another. Kinesin, a type of molecular motor, is known to transport cargo in discrete steps of $8\ nm$ (nano meter) predominantly, also $4 \ nm $ and $12 \ nm$ occasionally \cite{yildiz2004kinesin}. Due to physical scale of operation, the motion of kinesin takes place in the presence of Brownian motion, and hence, any observation of the discrete kinesin displacement is corrupted by a Gaussian; independent of the step size of kinesin as shown in Figure \ref{fig:kinesinsteps}. Another example is unfolding events of the domains of proteins, which are discrete events, but the force (in the pico-Newton range) at which the various domains unfold are not deterministic due to the influence of Brownian motion \cite{carrion1999mechanical}. The above two applications have challenging signal processing and inference problems of relevance to domain scientists and hence, entropy of $X+Z$ is an important quantity to understand.        

\begin{figure}
	\centering
	\includegraphics[scale = 0.3]{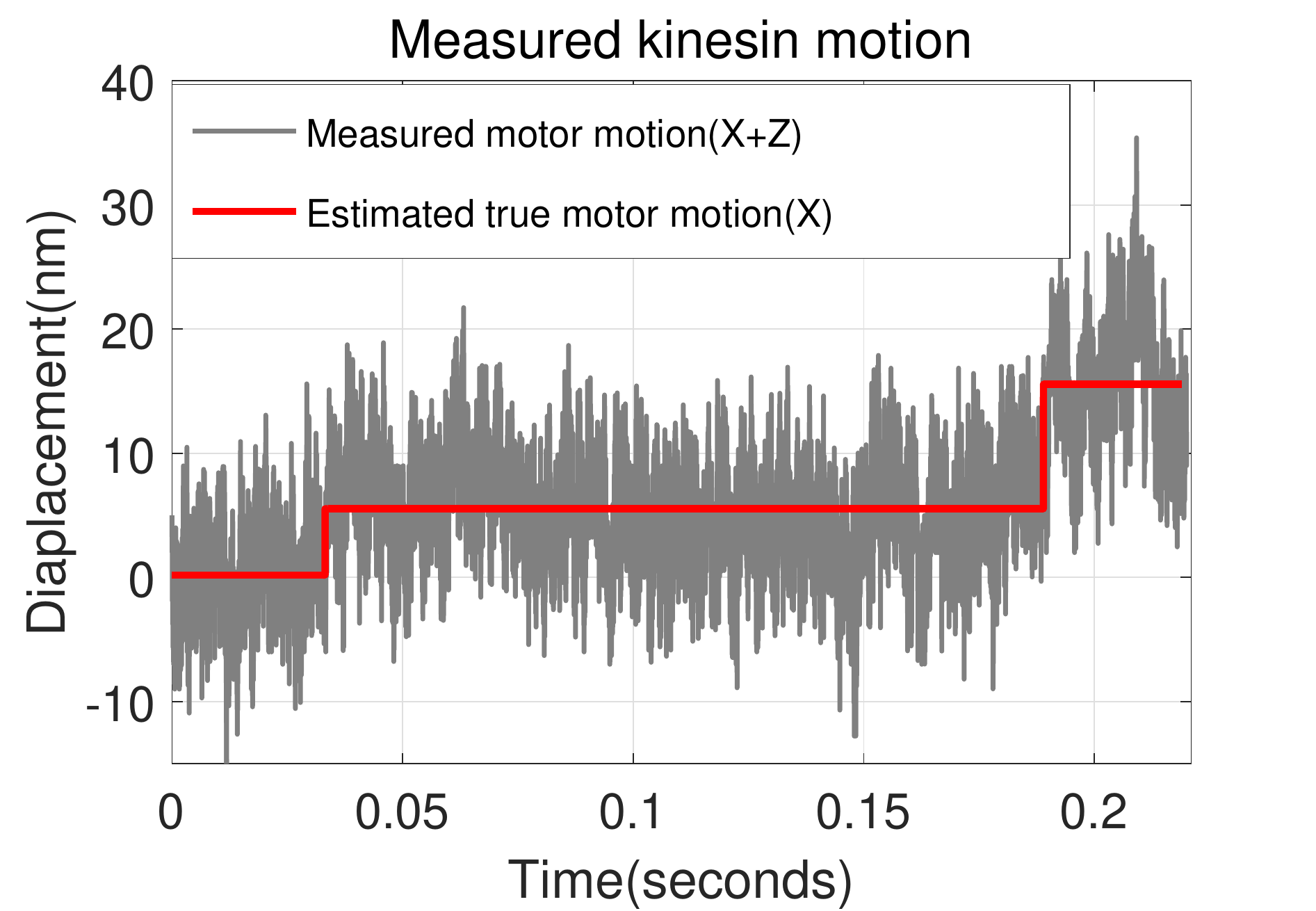} 
	\vspace{-0.2 cm}
	\caption{Experimentally measured motion of kinesin(grey) and the estimated true motion of kinesin(red) inferred from the measured motion \cite{aggarwal2012detection}.}
	\label{fig:kinesinsteps}
\end{figure}

The abstraction of a mixed random variable $X+Z$ has intimate connections with the state of a single bit memory and its entropy has fundamental links to Information Theory. In this case, $Z$ has two possible values, and, hence, is a Bernoulli random variable. The physical dimension of a single bit memory is in the nanometer regime, where thermal fluctuations (Brownian motion) play a key role in the device physics. The most commonly used description of the physics of a single bit memory is a particle in a double well potential, where, a barrier separates the two wells as shown in Figure \ref{fig:symmetric_well}, under the influence of Brownian motion. If the particle is in the right/left well the state of the memory can be considered to be one/zero respectively.  Most often, the probability distribution of the particle in either well is given by a Gaussian distribution \cite{chiuchiu2015conditional}. Henceforth, we assume that the probability distributions of the particle in the left and right well are $f_0(x):=\mathcal{N}(-\mu,\sigma^2)$ and $f_1(x):=\mathcal{N}(\mu,\sigma^2)$ respectively, where $\mathcal{N}(\mu,\sigma^2)$ denotes a Normal distribution function with mean $\mu$ and variance $\sigma^2$. It is equally likely for the discrete variable, $Z$, to be zero or one, that is, $P(Z=0) = P(Z=1)=\frac{1}{2}$ and $P$ denotes the probability measure. The probability of finding the Brownian particle between $x$ and $x+dx$ is given by, 
\begin{equation}\label{eq:pdf_i}
    \begin{aligned}
      &P(Z=0)P(X\in (x,x+dx)|Z=0) \\
        &+ P(Z=1)P(X\in (x,x+dx)|Z=1)\\
         &= \frac{1}{2}f_0(x)dx + \frac{1}{2}f_1(x)dx.
         \end{aligned}     
\end{equation} 
 Thus, the probability distribution function of the particle representing a single bit memory, $f_s(x)$ is an equally weighted mixture of $f_0(x)$ and $f_1(x)$. Of particular interest is the reset operation of a bit, where, irrespective of the information stored in the memory is one or zero, the outcome is zero.  Thus, after applying a reset protocol to the memory bit, $P(Z=0)=1, P(Z=1)=0$. Then, the probability density function of the particle after undergoing the reset operation, $f_e(x)=f_0(x)$. It is seen that there is a decrease in entropy(thermodynamic as well as Shannon) of a memory bit when undergoing a reset operation. This necessitates dissipation of heat, which results in the Landauer's principle linking information processing with thermodynamic costs \cite{landauer1961irreversibility}. It states that successful reset of $1$ bit of information stored in a memory is always accompanied by $k_BT\ln{2}$ amount of heat dissipation. Understanding of Landauer's principle involves computation of entropy differences between $f_s(x)$ and $f_e(x)$, which are usually modeled as Gaussian mixtures.

\begin{figure}
	\centering
	\includegraphics[scale = 0.3]{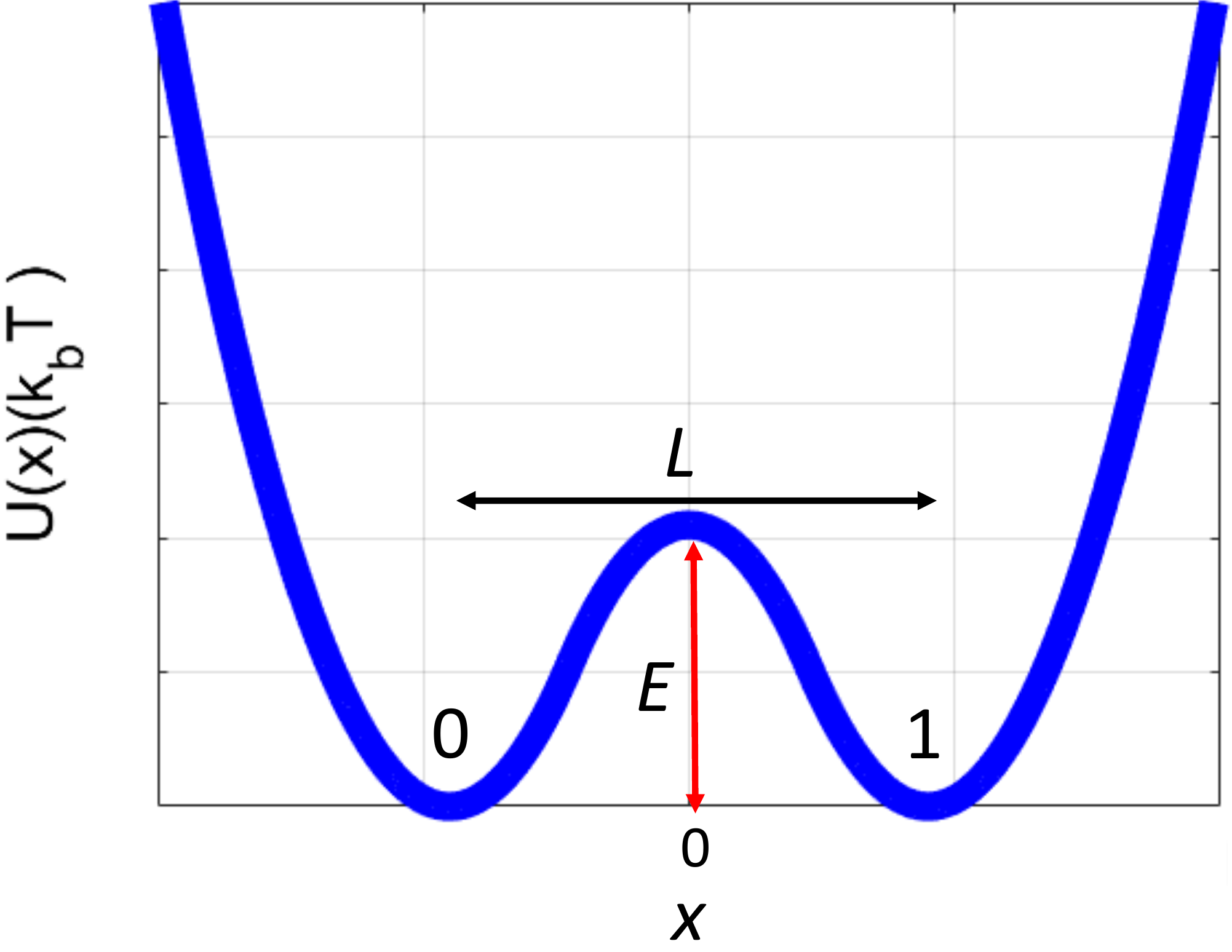} 
	\vspace{-0.2 cm}
	\caption{Schematic for a double well potential associated with a sinle bit memory. The x-axis denotes the position and the y-axis denotes the energy of the particle. The two wells are separated by distance $L$ and a barrier height of $E$. Particle in right well ($x>0$) represents logic state $1$ and particle in left well ($x<0$) represents logic state $0$ \cite{talukdar2017beating}.}
	\label{fig:symmetric_well}
	\vspace{-0.4 cm}
\end{figure}


Motivated from the above discussion and due to the lack of analytical expression for Gaussian mixtures, in this article we will study bounds on the entropy of $X+Z$, where $Z$ is a discrete random variable (not necessarily Bernoulli random variable) and $X$ being a Gaussian random variable independent of $Z$. Due to the pervasiveness of such distributions, and the usefulness of understanding their entropy, there is significant but disjoint literature on the topic. The interested reader can find related investigations in \cite{moshksar2016arbitrarily,huber2008entropy, kolchinsky2017estimating, michalowicz2008calculation}. We believe the interpretation of this problem as the deficit in an entropic inequality to be novel. 

\subsection{Contribution}
This article presents sharp bounds on the entropy of the sum of two independent random variables, $h(X+Z)$, with $X$ being a Gaussian random variable independent of $Z$ with $Z$ being a discrete random variable. While $H(X) + h(Z)$ is a trivial upper bound for $h(X+Z)$ (see for example \cite{WM2014beyond} where it is an immediate corollary of Lemma $11.2$), where, $H(Z)$ is the discrete entropy and $h(X)$ denote the entropy of a continuous random variable, our efforts sharpen this upper bound in the case of $X$ being a Gaussian and $Z$ being a discrete random variable independent of $X$. In particular, we explicitly characterize the gap between $h(X+Z)$ and $H(X)+h(Z)$. The application of these bounds in the case of thermodynamics of resetting a bit can be found in \cite{talukdar2018analyzing}. 

\subsection{Organization}
In the next section we present some definitions and preliminaries for the discussion of bounds on $h(X+Z)$, following which, in Section III, we present the bounds on $h(X+Z)$ when $X$ is Gaussian. In Section IV we show that these derived bounds are sharp, followed by Conclusion and Future Work in Section V and VI respectively. 

\section{Background}
We first present the notion of information entropy for discrete and continuous random variables.  The reader can consult \cite{CT91:book} for general background on information theory, and \cite{MMX17:1} for recent developments in entropic inequalities.
\begin{defn}
For an integer valued random variable $Z$  with the probability mass function, $P(Z = i)=p_i$, we denote the usual Shannon entropy in ``nats'' as,
\begin{equation} \label{eq:discrete entropy}
    H(Z) = - \sum_{i=-\infty}^{\infty} p_i \ln p_i.
\end{equation}

For a random variable $X$ with density $f(x)$ on $\mathbb{R}$,  whenever $f \ln f \in L^1$, we denote the entropy in the usual manner as,
\begin{equation} \label{eq:continuous entropy}
    h(X) = - \int_{-\infty}^{\infty} f(x) \ln f(x) dx.
\end{equation}
\end{defn}

\noindent In this article $\displaystyle \sum_{i=-\infty}^{\infty}a_i$ will be denoted as $ \displaystyle \sum_{i}a_i$, $\displaystyle \sum_{i=-\infty,i\neq 0}^{\infty}b_i$ as $\displaystyle \sum_{i\neq 0}b_i$. We will suppress notation at times when the meaning of expressions is clear from context.  We utilize $P (A)$ to denote the probability of an event $A$.

First a comment on the nature of $X+Z$ for independent $Z \sim p$ and $X \sim f$, $X+Z \sim f_{X+Z}$. As mentioned, when $X$ is Gaussian, $X+Z$ can be interpreted as a Gaussian mixture, but so long as $X$ has a density, $X+Z$ does as well. 
\[
    f_{X+Z} (x) = \sum_{k} p_k f(x-k),
\]
as for a Borel set $A$,
\begin{equation*}
    \begin{split}
         \int_A \sum_{k} p_k f(x-k) dx
            &=
                \sum_k p_k \int_A f(x-k) dx
                    \\
            &=
                \sum_k P(Z=k) P(X \in -k + A)
                    \\
            &=
                \sum_k P(Z=k, X \in -k +A)
                    \\
            &=
                P(X+Z \in A).
    \end{split}
\end{equation*}
 Thus, the notation $h(X+Z)$ is well defined in the following Proposition.

\begin{prop} \label{prop:mixed entropy inequality} For a $\mathbb{Z}$ valued random variable $Z$, and $X$ independent of $Z$ and taking values in $\mathbb{R}$,
\begin{align}\label{eqn:h(X+Z)}
    h(X+Z) = H(Z) + h(X) - \delta(X,Z),
\end{align}
where,
\begin{equation} \label{eq: delta definition}
\begin{split}
    \delta(&X,Z)  
        =  
            \\
        &\sum_k p_k \int f(x-k) \ln \left( 1 + \frac{ \sum_{j \neq k} p_j f(x - j)}{p_k f(x-k)} \right) dx,
\end{split}
\end{equation}
and satisfies
\begin{align} \label{eq: delta bigger than 0}
    \delta(X,Z) \geq 0.
\end{align}
\end{prop}

We take by convention $0 \ln 0$ as its continuous limit $0$, and implicitly consider the integral in the computation of $h(X+Z)$ to be taken only over $x$ such that $f_{X+Z}(x) \ln f_{X+Z}(x) > 0$.  The careful reader will notice that this precludes the possibility of division by zero in what follows. 

\begin{proof}

The entropy of $X+Z$ can be computed as follows,
\begin{align}
    h(X+Z) 
        =
            - &\int f_{X+Z} \ln f_{X+Z} dx,\nonumber
                \\
        =
            - &\int \sum_k p_k f(x-k) \ln \left( \sum_j p_j f(x-j) \right) dx,\nonumber
                \\
        =
            - & \sum_k p_k \int f(x-k) \ln f(x-k) dx ,\nonumber
                \\
         &-\sum_k p_k \ln p_k - \delta(X,Z),
\end{align}
where, using the translation invariance of entropy 
\[
    h(X)=-\int f(x-k)\ln f(x-k)dx,\text{for any}\ k.
\] 
 Thus, it follows that  \eqref{eqn:h(X+Z)} is satisfied.  Equation \eqref{eq: delta bigger than 0} follows immediately from \eqref{eq: delta definition}.
\end{proof}
\begin{rmk}
The inequality in the above Proposition is trivially sharp when one takes $Z$ to be a point mass. A slightly more interesting case of equality is when $X$ is supported in the interval $(-\frac 1 2, \frac 1 2)$.
\end{rmk}
 We now set out to bound $\delta (X,Z)$, and we begin with the substitution $y = x-k$, so that our error term, $\delta(X,Z)$ is expressed as, 
\begin{align}\label{eq:A delta bound inside of Remark}
     \sum_k p_k \int f(y) \ln \left( 1 + \frac{ \sum_{j \neq k} p_j f(y+k - j)}{p_k f(y)} \right) dy.
\end{align}
Bringing the sum inside the integral
and
applying Jensen's inequality 
to the concavity of logarithm we have for every $y \in \mathbb{R}$,
\begin{equation*}
    \begin{split}
       \sum_k p_k \ln \left( 1 + \frac{ \sum_{j \neq k} p_j f(y+k - j)}{p_k f(y)} \right)
            \\
            \leq
                 \ln \left( 1 + \frac{ \sum_k p_k \sum_{j \neq k} p_j f(y+k - j)}{p_k f(y)} \right),
    \end{split}
\end{equation*}
    
So that \eqref{eq:A delta bound inside of Remark} can be bounded above by
\begin{equation*}
    \begin{split}
            \int f(y) &\ln \left( 1 + \frac{ \sum_{k} \sum_{j \neq k} p_j f(y+k-j)}{f(y)} \right) dy. 
    \end{split}
\end{equation*}
Changing the order of summation leads to,
\begin{align*}
             \int f(y) \ln \left( 1 + \frac{ \sum_{m \neq 0}  f(y+m)}{f(y)} \right) dy.
\end{align*}
Let us collect these observations, which hold in generality in the following.

\begin{lem} \label{lem:general error bound}
    If $Z \sim p$ is a $\mathbb{Z}$ valued random variable and $X \sim f$ is a continuous random variable with $X$ and $Z$ independent and having bounded entropy, then the error term, 
    \[
        \delta(X,Z) := H(Z) + h(X) -h(X+Z),
    \]
    in addition to being non-negative satisfies,
    \[
        \delta(X,Z) \leq \int f(y) \ln \left( 1 + \frac{ \sum_{m \neq 0}  f(y+m)}{f(y)} \right) dy.
    \]
\end{lem}

We would like to evaluate this deficit $\delta(X,Z)$, when $Z$ is a Gaussian and $X$ is a discrete random variable independent of $Z$. 

\section{Bounds in the Gaussian Case}

Here we derive explicit bounds on $\delta(X,Z) $, when $X$ is a mean zero Gaussian random variable with $\sigma \leq 1/2$ and density,
\begin{align}\label{eq:gaussiandensity}
f(x)=\frac{1}{\sqrt{2\pi\sigma^2}}e^{-x^2/2\sigma^2}.
\end{align}
Under the assumption $\sigma < 1/2$ we will show that $h(X)+H(Z)$ is a \lq good\rq \ approximation of $h(X+Z)$. We will show in the next section, that this approximation is \lq poor\rq \ when $\sigma \geq 1/2$.  The bounds are derived by splitting the integral into two pieces.  We bound the segment with $y$ close to zero in Lemma \ref{lem:bounds close to zero}, for large $y$, we will make use of the following lemma.

\begin{lem} \label{lem:Sum of integer points of a log-concave distribution}
If $f$ is the Gaussian density as described in (\ref{eq:gaussiandensity}) with $\sigma \leq \frac 1 2$, then, 
\[
    \sum_{m \in \mathbb{Z}} f(\varepsilon + m) < \frac 1 \sigma,
\]
for any $\varepsilon \in \mathbb{R}$.
\end{lem}

\begin{proof}
    Without loss of generality\footnote{Since $\varepsilon = n + \varepsilon'$ for some integer $n$ and $|\varepsilon'| \leq 1/2$, a change of parameters allows us to assume $|\varepsilon| \leq 1/2$, and by the symmetry of $f$ it follows that 
    $
        \sum_m f(\varepsilon +m) = \sum_m f(-\varepsilon + m).$
    }, let $0 \leq \varepsilon \leq \frac 1 2$.   Straight forward computations using an upper-bounding geometric series and some numerical approximations will then achieve the lemma. 
    For $n \geq 0$ it is straightforward to obtain, 
    $
        e^{-(\varepsilon + n+1)^2/2\sigma^2} \leq e^{-2} e^{-(\varepsilon + n)^2/2 \sigma^2}.
    $
    By iterating this result, and using a geometric series bound along with $e \geq 2.7$, we obtain,
    \[
        \sum_{n=0}^\infty f(\varepsilon + n) \leq \frac 7 6 e^{-\varepsilon^2/\sigma^2}/\sqrt{2 \pi \sigma^2}.
    \]
    By noticing the initial indices of the two summations we have
    \[
        \sum_{k=1}^\infty f(\varepsilon -k) \leq \sum_{n=0}^\infty f(\varepsilon + n),
    \]
    which implies,
    \[
        \sum_{m} f(\varepsilon + m) \leq \frac 7 3 e^{-\varepsilon^2/\sigma^2}/\sqrt{2 \pi \sigma^2} \leq \frac 7{3 \sigma \sqrt{2 \pi}}.
    \]
    Using the approximation $\pi \geq 3$, we obtain,
    \[
         \sum_{m} f(\varepsilon + m) \leq \frac 1 \sigma.
    \]
    
\end{proof}

\begin{lem} \label{lem:bounds close to zero}
For a continuous random variable $X$ with density $f$, satisfying $f(-x)=f(x)$, 
\begin{equation*}
    \begin{split}
\int_{-\frac{1} 2}^{\frac 1 2} f(y) \ln \left( 1 + \frac{ \sum_{m \neq 0}  f(y+m)}{f(y)} \right) dy \\
\leq  2 \int_{\frac 1 2}^\infty f(y) dy.
    \end{split}
\end{equation*}
\end{lem}

\begin{proof}
    Using the inequality $\ln (1+x) \leq x$ for $x \geq 0$, we have,
    \begin{equation*}
        \begin{split}
          \int_{-\frac 1 2}^{\frac 1 2} f(y) \ln \left( 1 +  \frac{  \sum_{m \neq 0} f(y + m)}{f(y)} \right) dy
            \\
                \leq
           \int_{-\frac 1 2}^{\frac 1 2} \sum_{m \neq 0} f(y + m)  dy.
        \end{split}
    \end{equation*}
    After exchanging the summation and integral, the right hand side is,
    \[
        \int_{\{|y|> \frac 1 2 \}} f(y) dy.
    \]
    By using symmetry of the density, we have our result.
\end{proof}

\begin{lem} \label{lem: bounds away from zero}
    If $f$ is the Gaussian density as described in (\ref{eq:gaussiandensity}) with $\sigma < \frac 1 2$, then,
    \begin{equation*}
    \begin{split}
        \int_{|y| > \frac 1 2} f(y) \ln \left( 1 +  \frac{  \sum_{m \neq 0} f(y + m)}{f(y)} \right) dy \\ \leq \frac{f(1/2)}{2} + 5 \int_{\frac 1 2}^\infty f(y) dy.
    \end{split}
    \end{equation*}

\end{lem}

\begin{proof}
    Using Lemma \ref{lem:Sum of integer points of a log-concave distribution}, to bound $\sum_{m \neq 0} f(y+m)/f(y)$, gives 
    \begin{equation*}
        \begin{split}
                \int_{|y| > \frac 1 2} f(y) \ln \left(\frac{  \sum_{m \neq 0} f(y + m)}{f(y)} \right) dy
                    \\
            \leq
                \int_{|y| > \frac 1 2} f(y) \ln \left( \sqrt{2 \pi} e^{y^2/2 \sigma^2} \right) dy.
        \end{split}
    \end{equation*}
    After integrating by parts the right hand side of the above can be computed exactly as,
    \begin{equation*}
        \begin{split}
                &2 \int_{ \frac 1 2}^\infty f(y) \left(\ln \sqrt{ 2 \pi} + \frac{y^2}{2 \sigma^2}\right) dy
                    \\
            &=
                \frac{f(1/2)}{2} +  \int_{\frac 1 2}^\infty (1 + 2 \ln \sqrt{2 \pi})f(y) dy.
        \end{split}
    \end{equation*}
    Compiling all of the above and using the numerical approximation $1 + \ln 2 \pi+2\ln 2 \leq 5$ we have our result.
\end{proof}

\begin{thm} \label{thm:main result}
    If $Z$ is an integer valued random variable and $X$ is an independent Gaussian with density as described in (\ref{eq:gaussiandensity}) with $\sigma < \frac 1 2$, then, 
    \[
        0 \leq h(Z) + H(X) - h(X+Z) \leq \frac{e^{-1/8\sigma^2}} {\sqrt{2\pi}} \left( \frac 1 {2\sigma} + 7 \right).
    \]
\end{thm}

\begin{proof}
    As we have seen in Lemma \ref{lem:general error bound}, $\delta(X,Z)$ defined as, 
    \[
        \delta(X,Z) = H(X) + h(Z) - h(X+Z),
    \]
    is non-negative and satisfies, 
    \[
        \delta(X,Z) \leq  \int f(y) \ln \left( 1 + \frac{ \sum_{m \neq 0}  f(y+m)}{f(y)} \right) dy.
    \]
    We can bound the right hand side above by splitting the integral into two pieces and applying Lemmas \ref{lem:bounds close to zero} and \ref{lem: bounds away from zero} as follows,  
    \begin{align} \label{eq:expression as Gaussian tails}
            \int_{-\frac{1}{2}}^{\frac{1}{2}} &+ \int_{|y| > \frac 1 2} f(y) \ln \left( 1 + \frac{ \sum_{m \neq 0}  f(y+m)}{f(y)} \right) dy
               \nonumber, \\
            &\leq 
                2 \int_{\frac 1 2}^\infty f(y) dy + \left( \frac{f(1/2)}{2} + 5 \int_{\frac{1}{2}}^\infty f(y) dy\right)
                  \nonumber,  \\
            &=
                \frac{f(1/2)}{2} + 7 \int_{\frac{1}{2}}^\infty f(y) dy.
    \end{align}
    Using the substitution $w = y/\sigma$ (observe that the assumption $\sigma < \frac 1 2$ ensures $ \frac 1 {2 \sigma} > 1$) we have,
    \begin{align*}
        \sqrt{2 \pi} \int_{\frac 1 2}^\infty f(y) dy 
            &=
                \int_{\frac 1 {2\sigma}}^\infty e^{-w^2/2} dw, 
                     \\
            &\leq
                \int_{\frac 1 {2\sigma}}^\infty w e^{-w^2/2} dw, 
                    \\
            &=
                e^{-1/8\sigma^2}.
    \end{align*}
   Moreover, $f(1/2) = e^{-1/8\sigma^2}/\sqrt{2\pi} \sigma$. Combining the above computations with (\ref{eq:expression as Gaussian tails}) we have,
    \[
        \delta(X,Z) \leq e^{-1/8 \sigma^2} \left( \frac{(2\sigma)^{-1} + 7}{ \sqrt{2 \pi}}  \right).
    \]
\end{proof}

 \section{Sharpness of Bounds} \label{sec: Sharpness of bounds}
 We now show that the bound derived in the previous section is tight and cannot be improved significantly.  Consider the discrete random variable $Z$ to be a Bernoulli$(1/2)$ which we denote by $B$, then,
 \[
    f_{X+B}(x) = \frac 1 {\sqrt{8 \pi \sigma^2}} \left( e^{-x^2/2\sigma^2} + e^{-(x-1)^2/2 \sigma^2}\right).
 \]
 
 Using the symmetry of the Gaussian, $\delta(B,Z)$ reduces to the following in the fair Bernoulli case.
 
 \begin{equation*}
    \begin{split}
    H(B) + &h(X) - h(X+B)
        \\
        =&
            \int_\mathbb{R} \frac{e^{-x^2/2 \sigma^2}}{\sqrt{2 \pi \sigma^2}} \ln \left( 1 +  e^{\frac{x}{\sigma^2} -\frac{1}{2 \sigma^2}} \right) dx.
     \end{split}
 \end{equation*}
 Substituting $y = x /\sigma$, we obtain the following expression which is immediately bounded.
 \begin{equation} \label{eq: big sigma bad}
    \int_\mathbb{R} \frac{e^{-y^2/2}}{\sqrt{2 \pi}} \ln \left( 1 + e^{\frac y \sigma - \frac 1 {2 \sigma^2}}\right) dy
 \geq 
    \ln(2) \int_{\frac 1 {2\sigma}}^\infty \frac{e^{-y^2/2}}{\sqrt{2\pi}} dy.
 \end{equation}
 
 Observe that in the case $\sigma \geq \frac 1 2$ this demonstrates a lower bound on the error growing with $\sigma$.  In particular, $X=B$ is an example of \lq large deficit\rq \ for all $\sigma \geq \frac 1 2$.  
 
 We proceed forward with the assumption that $\sigma < \frac 1 2$, and use conventional bounds on Gaussian tails to show the sharpness of our upper bounds on $\delta(X,Z)$.
 
 Recall for the standard normal, $f(y) = \frac{e^{-y^2/2}}{\sqrt{2\pi}}$ when $z >0$
 \begin{equation} \label{eq:Lower bounds on Gaussian tails}
    \int_z^\infty f(y) dy \geq f(z)\left(\frac 1 z - \frac 1 {z^3}\right).
 \end{equation}
 This follows from the equation, $f(y) = - f'(y)/y$ and application of integration by parts twice.  Applying \eqref{eq:Lower bounds on Gaussian tails} with $z = \frac{1}{2 \sigma}$ we have,
 \begin{equation} \label{eq:Bernoulli lower bounds}
    \delta(X,B) \geq \frac{e^{-1/8\sigma^2}}{\sqrt{2\pi}} \left( \ln(2)\left( 2 \sigma  - 8\sigma^3 \right) \right).
 \end{equation}

 \section{Conclusion}
 We have shown the deficit in the inequality $h(X+Z) \leq H(Z)+h(X)$ has Gaussian decay with $1/\sigma$.  What is more,
up to the polynomial and rational terms in $\sigma$, the lower bound on $\delta(B,Z)$ derived in equation \eqref{eq:Bernoulli lower bounds} matches the upper bounds derived for general $\delta(X,Z)$.  As such, Theorem \ref{thm:main result} can be considered sharp with small scope of improvement.  Additionally \eqref{eq: big sigma bad} gives a quantitative example of large deficit when $\sigma \geq \frac 1 2$.

 \section{Future Work}
 We remark that the placement of the discrete random variable on $\Z$ is more a product of convenience than necessity. Similar derivations are possible for more general discrete variable on $\R$,  the distance between  values of $Z$ relative to the strength of the noise $X$ will remain pertinent. 
 It is of interest to study the error term in the case of non Gaussian random variables, in particular to attempt the generalization of these results to log-concave random variables. Extensions and further applications of the results here will be the topic of a subsequent article \cite{Melbourne2018logconcavedeficit}.
 
\section{Acknowledgement}
 We would like to thank Prof. Mokshay Madiman, University of Delaware and Prof. Arnab Sen, University of Minnesota for initial discussion about the problem. The authors acknowledge the support of the National Science Foundation for funding the research under Grant No. CMMI-1462862, CNS 1544721 and the first author acknowledges support from NSF Grant No. 1248100.  Portions of this article and related results were announced at the March Meeting of the American Physical Society  \cite{talukdar2018mixture, melbourne2018mixed}.
 
\bibliographystyle{plain}
\bibliography{bibibi}

\end{document}